\documentclass[aps,11pt,twoside]{revtex4}
\usepackage[latin1]{inputenc}
\usepackage{amsmath,amsbsy,amsfonts,hyperref,bbm,verbatim,amsthm,amssymb}

\newcommand{\ket}[1]{\left|#1\right\rangle}
\newcommand{\cket}[1]{\left|\widetilde{#1}\right\rangle}
\newcommand{\bra}[1]{\left\langle #1\right|}
\newcommand{\bracket}[2]{\left\langle #1|#2\right\rangle}
\newcommand{\encx}{{\bar{\bf{x}}}}
\newcommand{\synx}{{\hat{\bf{x}}}}
\newcommand{\encz}{{\bar{\bf{z}}}}
\newcommand{\synz}{{\hat{\bf{z}}}}
\newcommand{\hmax}[1]{{H}_{\max}^{#1}}
\newcommand{\hmin}[1]{{H}_{\min}^{#1}}

\newtheorem{theorem}{Theorem}

\begin{document}

\title[Duality of privacy amplification and data compression]{Duality of privacy amplification against quantum adversaries and data compression with quantum side information}
\author{Joseph M. Renes}
\affiliation{Institut f\"ur Angewandte Physik, Technische Universit\"at Darmstadt, Hochschulstr.~4a, 64289 Darmstadt, Germany}
\label{firstpage}

\begin{abstract}
We show that the tasks of privacy amplification against quantum adversaries and data compression with quantum side information are dual in the sense that the ability to perform one implies the ability to perform the other. These are two of the most important primitives in classical information theory, and are shown to be connected by complementarity and the uncertainty principle in the quantum setting. Applications include a new uncertainty principle formulated in terms of smooth min- and max-entropies, as well as new conditions for approximate quantum error correction. 
\end{abstract}

\maketitle
\section{Introduction}

Two of the most fundamental primitives in information theory are privacy amplification and data compression with side information, both of which involve manipulating the correlations between two random variables $Z$ and $Y$. Privacy amplification is the art of extracting that part of $Z$ which is uncorrelated from $Y$. In particular, the goal is to extract uniform randomness, in the form of a random variable $U$, from an input $Z$ in such a way that $U$ is completely uncorrelated with $Y$. In a cryptographic setting $Z$ might refer to a partially-random and partially-private key string, while $Y$ refers to knowledge held by an adversary. Meanwhile, the goal of data compression with side information is essentially the opposite, to determine that part of $Z$ which is not correlated with $Y$ and to make this available as the compressed version of $Z$. More specifically, an encoder would like to compress $Z$ into a smaller random variable $C$ such that a decoder can reconstruct $Z$ given both $C$ and the side information $Y$. 

These two tasks have direct, purely quantum analogs in quantum information theory. Data compression with side information translates into distillation of entanglement between two quantum systems $A$ and $B$ using only classical communication (the analog of $C$). The quantum version of privacy amplification is the removal of all correlations (both quantum and classical) between $A$ and $B$ by actions taken only on $A$, such that the density matrix for system $A$ is also transformed into a completely mixed state (the analog of $U$). 

Moreover, in the purely quantum realm the two quantum tasks are dual to one another, a feature which has been fruitfully exploited to construct a whole family of quantum information processing protocols~\cite{abeyesinghe_mother_2009}. The duality holds for complementary quantum systems, in the sense that if it is possible to maximally entangle two systems $A$ and $B$ such that $A$ itself is maximally mixed, then it is possible to completely decouple a maximally-mixed $A$ from the complementary system $R$ of $B$, and vice versa~\cite{schumacher_approximate_2002}. Two systems $B$ and $R$ are complementary, relative to system $A$, when the joint quantum state of $ABR$ is a pure state, a state which always exists in principle. That is, two systems $B$ and $R$ are complementary relative to $A$ when each one is the purifying system for the other and $A$~\footnote{In the communication scenario complementary systems become complementary channels. An arbitrary channel $\mathcal{E}$ taking $A$ to $B$ can be thought of as an isometry taking $A$ to $BR$ by the Stinespring dilation, and considering only the $R$ portion of the output defines the complementary channel $\mathcal{E}^\#$ to $\mathcal{E}$.}. 

In this paper we show that this duality also extends to the hybrid classical-quantum tasks of classical privacy amplification against quantum adversaries and classical data compression with quantum side information: The ability to perform one of the tasks implies the ability to perform the other. Here we are interested in manipulating the correlations between a classical random variable $Z$ and a quantum random variable, i.e.\ a quantum system $B$. Despite the hybrid nature of the resources, the analysis is still within the realm of quantum information theory, as we can and do imagine that $Z$ is produced by measurement of the quantum system $A$. 

Complementary quantum systems still constitute an important part of the duality, and compression of $Z$ given side information $B$ implies privacy amplification against $R$ and vice versa, just as in the purely quantum case. However, the duality takes on an additional complementary character, as the compression task is not dual to privacy amplification of $Z$ against $R$, but rather it is dual to privacy amplification of a complementary random variable, which we will call $X$, against $R$. Complementary random variables correspond to outcomes of measuring complementary observables, observables like position and momentum for which complete certainty in the outcome of one observable implies complete uncertainty in the outcome of the other. In the present context, if the random variable $Z$ is the result of measuring an observable $Z^A$ on system $A$, then $X$ is the result of measuring a complementary observable $X^A$ on $A$. In what follows we ignore the difference between an observable and random variable and simply call both $Z^A$ (or ${X}^A$). 

Of course, one of the pillars of quantum mechanics is that both measurements cannot be performed simultaneously. Because analysis of such complementary measurements can quickly become a maze of counterfactuals, let us describe the duality more precisely. We start with a pure quantum state $\psi^{ABR}$ describing the three quantum systems $A$, $B$, and $R$. Then we imagine a \emph{hypothetical} $Z^A$ measurement, say, and then design a protocol for data compression of the resulting classical random variable $Z^A$ given side information $B$. The protocol itself is real enough, and the duality then states that if we instead perform the $X^A$ measurement, it is possible to repurpose the compression protocol to perform privacy amplification of the classical random variable $X^A$ against system $R$. The same is true in the reverse direction (modulo the caveats discussed below). We stress that only one of the two conjugate measurements $Z^A$ or $X^A$ is ever performed on $\psi^{A}$; we merely contemplate what would be possible had the other measurement been performed.  


There are two caveats regarding the duality that should be emphasized. First, we can only establish a duality between protocols in which the privacy amplification function or data compression function is linear. This requirement stems from the need to interpret functions applied to $X^A$ as operations on $Z^A$ and vice versa. In general this is problematic, as $X^A$ and $Z^A$ are complementary observables and therefore actions taken on one have unavoidable back-actions on the other, but linear functions will offer a way around this problem. 

Secondly, the duality does not hold in both directions for arbitrary states of $ABR$. As we shall see, the ability to perform data compression with side information (CSI) implies the ability to perform privacy amplification (PA). However, we can only show the converse when $\psi^{ABR}$ has one of two simple forms, either $R$ is completely correlated with (a hypothetical measurement of) $Z^A$ or $B$ is completely correlated with (a hypothetical measurement of) $X^A$. These restrictions and the asymmetry of the duality can be traced back to a recently proven form of the  uncertainty principle~\cite{renes_conjectured_2009,berta_entropic_2009} and the fact that it only sets a \emph{lower} limit on knowledge of complementary observables. Going from privacy amplification to data compression implicitly requires an \emph{upper} limit, which we deal with by considering the equality conditions of the uncertainty principle, and these are shown to be exactly the two conditions named above. 


The remainder of the paper is devoted to elucidating the duality. In the next section we provide background on the two tasks, how protocols can be constructed using universal hashing, as well as the details of one-shot protocols handling arbitrary inputs and rates that can be achieved in the case of asymptotically-many identical inputs. Then in Sec.~\ref{sec:perfect} we examine the ``perfect'' cases, that is, when $Z^A$ is perfectly recoverable from $B$ or $R$ is completely uncorrelated with a uniformly random $X^A$, and show that the duality immediately follows from a recently discovered form of the uncertainty principle. Use of the uncertainty principle helps explain the duality in a simplified setting and understand the reason for the second caveat above. 

As perfect correlation or uncorrelation is difficult to achieve in practice, we are ultimately more interested in the approximate case.  In Sec.~\ref{sec:approx} we investigate the duality in the approximate case and show that $R$ is approximately uncorrelated with $X^A$ if $Z^A$ is approximately recoverable from $B$, and vice versa. This serves as a stepping stone to  studying full-fledged CSI and PA protocols, as taken up in Sec.~\ref{sec:iidprotocols}. Therein we show how CSI protocols utilizing linear hashing can be used to construct, and can be constructed from, linear-hashing based PA protocols. 
In the case of protocols designed for inputs consisting of asymptotically-many copies of a fixed resource state, the uncertainty principle of~\cite{renes_conjectured_2009,berta_entropic_2009} implies that the duality preserves optimality of protocols in that optimal CSI protocols can be transformed into optimal PA protocols, and vice versa. 
Combining this with recent results on one-shot CSI, this construction implies a new uncertainty principle formulated in terms of smooth min- and max-entropies, which we discuss in Sec.~\ref{sec:applications} along with additional applications and relations to other work.

\section{Background}
\subsection{Classical-Quantum States}
In order to describe protocols involving hybrid classical-quantum systems, it is convenient to work within the formalism of quantum mechanics. In this language, a classical random variable $Z^A$ and quantum side information $S$ can be described by the classical-quantum (cq) state
\begin{align}
\label{eq:cqstate}
{\psi}^{AS}_Z=\sum_{z=0}^{d-1} p_z\ket{z}\!\bra{z}^A\otimes\varphi_z^S,
\end{align}
where $z$ are the possible values the random variable $Z^A$ can take, with alphabet size $d$, $p_z$ is the probability that $Z^A=z$, and $\varphi_z^S$ is the quantum state of $S$  conditioned on the value of $z$. The $A$ system is effectively classical in the sense that an an unconditional measurement in the $\ket{z}$ basis has no effect on the state; essentially it has already been measured. The measurement defines the $Z^A$ observable, up to specifying the values of the possible outcomes, i.e.\ the position of the position observable. In the present context these values are irrelevant, as we are content with simply enumerating them. The subscript, here $Z$, indicates this is a cq state and which observable defines the classical basis. 

The entropy of the classical random variable $Z^A$ given the quantum side information $S$ is defined by 
\begin{align}
\label{eq:condent}
H(Z^A|R)_{{\psi}^{AS}_Z}\equiv H(AR)_{{\psi}^{AS}_Z}-H(S)_{{\psi}^{S}_Z},
\end{align}
where $H(A)_{\psi^A}=-{\rm Tr}\left[\psi^A\log_2 \psi^A\right]$ is the von Neumann entropy (measured in bits) and $\psi^S_Z={\rm Tr}_A[\psi^{AS}_Z]$ is the partial trace over the $A$ system.

In general, $\psi^{AS}$ can be thought of as the marginal of the pure state $\ket{\psi}^{AST}$,
where 
\begin{align}
\ket{\psi}^{AST}=\sum_{z=0}^{d-1}\sqrt{p_z}\ket{z}^A\ket{z}^{T_1}\ket{\varphi_z}^{ST_2}.
\end{align}
System $T$ consists of two parts, $T_2$ which purifies $S$ for each value of $z$ and $T_1$ which purifies $AST_2$. Here we name the systems $S$ and $T$ instead of $B$ and $R$ as in the introduction because in the subsequent sections $B$ and $R$ will take on both roles in different contexts. 

From the pure state we can still define the entropy of the classical random variable $Z^A$ given $S$ by first converting back to a cq state. We will often make use of the following definition: $H(Z^A|S)_{\psi^{AST}}\equiv H(Z^A|S)_{\psi^{AS}_Z}$. 

\subsection{Privacy Amplification Against Quantum Adversaries (PA)}
Privacy amplification is the art of extracting a truly secret random variable, a secret key, from one partially known to an eavesdropper holding some side information $S$. Functions performing privacy amplification are usually called \emph{extractors}, the goal being to produce as much secret key as possible. 
Privacy amplification against adversaries holding classical information was introduced in~\cite{bennett_to_1986,bennett_privacy_1988}, and was extended to the case of quantum side information in~\cite{devetak_distillation_2005,knig_power_2005,renner_universally_2005}. 

Using~(\ref{eq:cqstate}), the ideal output of privacy amplification would be a state for which $p_z=\frac{1}{d}$ and the $\varphi_z^S$ are all independent of $z$ and equal to one another. This last property implies that $\varphi_z^S=\psi^S$ for all $z$. 
In~\cite{renner_universally_2005} Renner and K\"onig introduced an approximate notion of security and uniformity of $Z^A$ which is universally composable, meaning that $Z^A$ can safely be input into any other composable protocol, and the overall protocol is certain to be secure by virtue of its constituent parts being secure. This definition says that $Z^A$ is approximately secure if the trace distance to the ideal case $\tfrac{1}{d}\mathbbm{1}\otimes {\psi}^S$ is small, where ${\psi}^S={\rm Tr}_A\left[{\psi}^{AS}_Z\right]$. We will say that $Z^A$ is $\epsilon$-secure if 
\begin{align}
\label{eq:epssecure}
p_{\rm secure}(X|S)_\psi\equiv\tfrac{1}{2}\left\|\psi_Z^{AS}-\tfrac{1}{d}\mathbbm{1}\otimes {\psi}^S\right\|_1\leq \epsilon.
\end{align}
Use of the trace distance $\left\|M\right\|_1\equiv{\rm Tr}[\sqrt{M^\dagger M}]$ means that the actual $\epsilon$-secure $Z^A$ can only be distinguished from the ideal, a perfect key, with probability at most $\epsilon$. 

Renner and K\"onig show that privacy amplification can produce an $\epsilon$-secure key of length (number of bits) $\ell_{\rm PA}^\epsilon(Z^A|S)_\psi$ given in terms of the smooth min-entropy~\cite{renner_universally_2005, renner_security_2005}: 
\begin{align}
\label{eq:pa}
\hmin{\epsilon_1}(Z^A|S)_\psi-2\log\tfrac{1}{\epsilon_2}+2\,\leq\, \ell_{\rm PA}^\epsilon(Z^A|S)_\psi\,\leq\, \hmin{\sqrt{2\epsilon}}(Z^A|S)_\psi,
\end{align}
where $\epsilon=\epsilon_1+\epsilon_2$. For a precise definition of the smooth min-entropy, see Appendix~\ref{app:smooth}. 

The lower bound is established by constructing an extractor based on \emph{universal hashing}. 
In this scheme the approximate key is created by applying a randomly chosen hash function $f$ to $Z^A$. The function is chosen from a universal family $F$ of hash functions, each mapping a size $d$ input alphabet to a size $m$ output alphabet, such that for any pair of inputs $z_1$ and $z_2$ the probability of a collision of the outputs is no greater than if the function were chosen at random from all functions:
\begin{align*}
{\rm Pr}_F[f(z_1)=f(z_2)]\leq \tfrac{1}{m}\qquad\forall\,\, z_1,z_2.
\end{align*}
More properly, such a family is called a 2-universal family, since the outputs exhibit a weaker form of pairwise independence. Hash families whose outputs are truly pairwise independent are called strongly 2-universal, a notion which can be easily extended to $k$-wise independence. In the present context we shall focus on using linear hash functions, and since the family of all linear hash functions is universal, we can immediately apply the results of~\cite{renner_universally_2005}. 

Meanwhile, the upper bound applies to any conceivable privacy amplification protocol, and stems from properties of the min-entropy itself. In the asymptotic i.i.d.\ case of $n\rightarrow \infty$ copies of $\bar{\psi}^{AS}_Z$, the min-entropy tends to the more well-known von Neumann entropy, $\hmin{\epsilon}(Z^{A^{\otimes n}}|S^{\otimes n})_{\psi^{\otimes n}}\rightarrow nH(Z^A|S)_\psi+O(\sqrt{n\log\frac{1}{\epsilon}})$~\cite{tomamichel_fully_2009}, which implies that in this case universal hashing can produce approximate keys at the rate
\begin{align}
\label{eq:parate}
r_{\rm PA}(\psi)\equiv\lim_{\epsilon\rightarrow 0}\lim_{n\rightarrow \infty}\tfrac{1}{n}\ell_{\rm PA}(Z^{A^{\otimes n}}|S^{\otimes n})_{\psi^{\otimes n}}=H(Z^A|S)_{\psi},
\end{align}
and furthermore this rate is optimal. 
These results nicely conform with the intuitive understanding of $\hmin{\epsilon}(Z^A|S)$ and $H(Z^A|S)$ as uncertainties of $Z^A$ given the side information $S$; the part of $Z^A$ unknown to $S$ should be roughly of this size, so it is sensible that this amount can in fact be extracted by privacy amplification.

\subsection{Data Compression with Quantum Side Information (CSI)}
The problem of data compression with side information, also known as information-reconciliation, is to compress a random variable $Z^A$ in such a way that it can later by recovered by the compressed version $Z'$ together with the side information $S$. Unlike privacy amplification, this protocol has two components, a compressor and a decompressor, the goal of course being to compress the input to as few bits as possible. The case of classical side information was first solved for in the asymptotic i.i.d.\ scenario by Slepian and Wolf~\cite{slepian_noiseless_1973}, and a one-shot version was given by Renner and Wolf~\cite{renner_simple_2005,renner_smooth_2004}. The quantum i.i.d.\ version was studied by Winter~\cite{winter_coding_1999} and Devetak and Winter~\cite{devetak_classical_2003}, and recently extended to the one-shot scenario by the present author~\cite{renes_one-shot_2010}.

The ideal output of such a scheme would be a cq state in which the $\varphi_z^S$ were perfectly distinguishable from one another, so that a corresponding measurement of $S$ would perfectly reconstruct $z$. A suitable approximate notion is that there should exist some measurement $\Lambda^S_z$ for which the probability $p_{\rm guess}(Z^A|B)$ of successfully identifying $z$ is large. When there does, we say $z$ is $\epsilon$-recoverable from $B$ in the sense that
\begin{align}
\label{eq:epsgood}
p_{\rm guess}(Z^A|B)=\sum_{z=0}^{d-1} p_z{\rm Tr}\left[\Lambda_z^S\varphi_z^S\right]\geq 1-\epsilon.
\end{align}

The one shot result can be formulated in terms of the dual quantity to the min-entropy, the max-entropy, defined in Appendix~\ref{app:smooth}. The minimum number of bits $\ell_{\rm CSI}^{\epsilon}(Z^A|S)_\psi$ needed to compress $Z^A$ such that it is $\epsilon$-recoverable from $S$ and the compressed version is bounded by
\begin{align}
\label{eq:csi}
\hmax{\sqrt{2\epsilon}}(Z^A|S)_\psi\,\leq\, \ell_{\rm CSI}^{\epsilon}(Z^A|S)_\psi\,\leq\, \hmax{{\epsilon_1}}(Z^A|S)_\psi+2\log\tfrac{1}{\epsilon_2}+4,
\end{align}
again for $\epsilon=\epsilon_1+\epsilon_2$. The upper bound is found by constructing a compressor using universal hashing and a decompressor using the so-called ``pretty good measurement'', while the lower bound follows from properties of the max-entropy. Like the min-entropy, the max-entropy also tends to the von Neumann entropy in the limit of $n\rightarrow\infty$ i.i.d.\ inputs. Defining the rate as in privacy amplification, we obtain
\begin{align}
\label{eq:csirate}
r_{\rm CSI}\equiv\lim_{\epsilon\rightarrow 0}\lim_{n\rightarrow \infty}\tfrac{1}{n}\ell_{\rm CSI}(Z^{A^{\otimes n}}|S^{\otimes n})_{\psi^{\otimes n}}=H(Z^A|S)_\psi.
\end{align}
This reproduces the results of Devetak and Winter for the asymptotic i.i.d.\ case.
Again this result conforms to the intuitive understanding of the conditional entropies. Since $\hmax{\epsilon}(Z^A|S)$ and $H(Z^A|S)_\psi$ are in some sense $S$'s uncertainty of $Z^A$, it is sensible that the compressor would have to augment the decompressor's information by this amount. 

\section{Duality from the Uncertainty Principle in the Perfect Case}
\label{sec:perfect}
We now show that the ideal cases that either $B$ is already perfectly correlated with $Z^A$ or $R$ is perfectly uncorrelated with $X^A$ are already dual by using a recently derived version of the uncertainty principle. Although using the uncertainty principle in this way will ultimately prove insufficient in the approximate case and when attemping to construct one protocol from the other, the analysis here serves to introduce the nature of the duality in a simplified setting, as well as understand the reasons behind the second caveat.

As remarked in the introduction, the duality between PA and CSI exists for complementary observables $Z^A$ and $X^A$. Let us be more specific and define these observables to be the Weyl-Heisenberg operators $Z\equiv\sum_{k=0}^{d-1} \omega^{k}\ket{k}\bra{k}$ and $X\equiv\sum_{k=0}^{d-1}\ket{k{+}1}\bra{k}$. Since they aren't Hermitian, these operators are not observables in the usual sense since the values they can take on are not real numbers. However, they each specify a basis of system $A$, enough to specify two measurements, which is is all we need here. The two are related by Fourier transform, since the eigenstates of $X$ are simply $\cket{x}=\frac{1}{\sqrt{d}}\sum_{z=0}^{d-1}\omega^{-xz}\ket{z}$. From this relation it is clear that the observables are complementary, as the result of $Z^A$ measurement on an $X^A$ eigenstate is completely random, and vice versa. 

Now consider a recently-discovered form of the uncertainty principle~\cite{renes_conjectured_2009,berta_entropic_2009}, which quantifies uncertainty by entropy and includes the case of quantum side information,
\begin{align}
\label{eq:cit}
H(X^A|R)_\psi+H(Z^A|B)_\psi\geq 1.
\end{align}
This holds for arbitrary states $\psi^{ABR}$, pure or mixed.
Loosely speaking, it states that the entropy $R$ has about the result of measuring $X^A$, plus the entropy $B$ has about the result of measuring $Z^A$, cannot be less than 1. Note that it is not possible to perform both of these measurements simultaneously, since the associated observables do not commute. Nevertheless, the uncertainty principle constrains what systems $B$ and $R$ can simultaneously ``know'' about the results of either measurement.

Let us see how this can be used to show that perfect $Z^A$ recovery from $B$ implies perfect $X^A$ security from $R$. Consider an arbitrary pure state $\ket{\psi}^{ABR}$, which we can write
\begin{align*}
\ket{\psi}^{ABR}&=\sum_z \sqrt{p_z}\ket{z}^A\ket{\varphi_z}^{BR}\\\nonumber
&=\sum_x \sqrt{q_x}\cket{x}^A\ket{\vartheta_x}^{BR}.
\end{align*}
using the $Z^A$ basis $\ket{z}^A$ or the $X^A$ basis $\cket{x}^A$. In the ideal case the states $\varphi_z^B$ are perfectly distinguishable, and therefore $H(Z^A|B)_\psi=0$. By the above this implies $H(X^A|R)_\psi=\log_2 d$, which can only occur if all the marginal states $\vartheta_x^R$ are identical. Hence $R$ is completely uncorrelated with $X^A$. Furthermore, $H(Z^A|B)_\psi=0$ also implies that $X^A$ is uniformly distributed. Since the uncertainty principle holds for any state, we can also apply it to $\psi^{AB}$. This yields $H(X^A)_\psi=\log_2 d$, meaning $X^A$ is uniformly distributed. Thus, $X^A$ is an ideal key, uniformly distributed and completely uncorrelated with $R$.

We cannot directly make use of the uncertainty principle for the converse, $X^A$ security from $R$ implies $Z^A$ recoverability from $B$. Assuming the former, we have $H(X^A|R)=1$. But this does not imply $H(Z^A|B)=0$ unless the uncertainty principle is tight. As an example, consider the $d=2$ the state $\ket{\psi}^{ABR}=\frac{1}{\sqrt{2}}\left(\ket{0}+i\ket{1}\right)^A\ket{\varphi}^{BR}$, for which $H(Z^A|B)=H(X^A|R)=1$. On the other hand, if the uncertainty principle is tight, then it is immediate that $H(X^A|R)=1$ implies $H(Z^A|B)=0$ and therefore the desired implication holds.

Thus we are interested in the equality conditions for Eq.~\ref{eq:cit}. The only currently known conditions are that (at least) one of $H(X^A|R)_\psi$, $H(X^A|B)_\psi$, $H(Z^A|R)_\psi$, or $H(Z^A|B)_\psi$ is zero~\cite{renes_conjectured_2009}, so that equality is met when the conditional entropies take on their extreme values. Put differently, the global state $\psi^{ABR}$ must in some way be a cq state, be it between $A$ and $R$, as in $\psi^{AR}=\psi^{AR}_X$ or $\psi^{AR}_Z$, or between $A$ and $B$, as in $\psi^{AB}=\psi^{AB}_Z$ or $\psi^{AB}_X$. Moreover, there must be perfect correlation between the two systems in the sense that the conditional marginal states in either $B$ or $R$ (which depend on the value of $X^A$ or $Z^A$) must be perfectly distinguishable.

For completeness, we briefly recapitulate the argument here. Consider the case that $H(Z^A|B)_\psi=0$, 
which immediately implies $H(X^A|R)_\psi\geq 1$. Since 1 is also an upper bound to the conditional entropy, it must be that $H(X^A|R)_\psi=1$ and the equality conditions are met. The same argument can be made starting from $H(X^A|R)_\psi=0$. The remaining two quantities $H(X^A|B)_\psi$ and $H(Z^A|R)_\psi$ are related by the complementary form of the uncertainty principle, obtained by interchanging either the complementary observables $X^A$ and $Z^A$ or the complementary systems $B$ and $R$. The derivation in~\cite{renes_conjectured_2009} simultaneously produces both forms of the uncertainty principle, meaning that satisfying the equality conditions for one implies the same for the other. Thus, the conditions $H(X^A|B)_\psi=0$ and $H(Z^A|R)_\psi=0$ also lead to equality in~(\ref{eq:cit}). 

Observe that in the former case of $H(X^A|B)=0$ and $H(X^A|R)=1$ we end up with $H(X^A|B)=H(Z^A|B)=0$, which is a sufficient condition to have maximal entanglement between $A$ and $B$, as discussed in~\cite{renes_physical_2008}. In the other case we end up with $H(Z^A|B)=H(Z^A|R)=0$ and $H(X^A|B)=H(X^A|R)=1$, a a situation similar to that of a $d=2$ GHZ state $\frac{1}{\sqrt{2}}\left(\ket{000}+\ket{111}\right)^{ABR}$.

\section{Duality in the Approximate Case}
\label{sec:approx}
In this section we examine the duality when $Z^A$ is approximately recoverable from $B$ or $R$ is approximately independent of a nearly uniform $X^A$. Unfortunately, the arguments using the uncertainty principle in the previous section cannot easily be modified to work in the approximate case, so here we present a more algebraic treatment. We start with $Z^A$ recovery implies $X^A$ security.

\begin{theorem}
\label{thm:recoverimpliessecure}
For an arbitrary pure state $\ket{\psi}^{ABR}$, suppose 
$p_{\rm guess}(Z^A|B)_\psi\geq 1-\epsilon$. Then $p_{\rm secure}(X^A|R)_\psi\leq \sqrt{2\epsilon}$. 
\end{theorem}
\begin{proof}
Start by performing the measurement coherently with a partial isometry $U^{B\rightarrow BM}$ and an ancillary system $M$. This transforms the state according to
\begin{align*}
\ket{\psi'}^{ABMR}&\equiv U^{B\rightarrow BM}\ket{\psi}^{ABR}\\
&\equiv\sum_{z,z'}\sqrt{p_z}\ket{z}^A\ket{z'}^M\sqrt{\Lambda_{z'}^B}\ket{\varphi_z}^{BR}.
\end{align*} 
The ideal output would be  
\begin{align*}
\ket{\xi}^{ABMR}=\sum_z\sqrt{p_z}\ket{z}^A\ket{z}^M\ket{\varphi_z}^{BR},
\end{align*}
and computing the fidelity $F(\psi',\xi)\equiv|{\bracket{\psi'}{\xi}}|$ between the two we find  
\begin{align*}
F(\psi',\xi)&=\bra{\xi}U^{B\rightarrow BM}\ket{\psi}^{ABR}\\
&=\sum_z p_z\bra{\varphi_z}\sqrt{\Lambda_{z}^B}\ket{\varphi_z}^{BR}
\\&\geq \sum_z p_z\bra{\varphi_z}{\Lambda_{z}^B}\ket{\varphi_z}^{BR}\\
&=p_{\rm guess}.
\end{align*}
Here we have used the fact that $\sqrt{\Lambda}\geq \Lambda$ for $0\leq \Lambda\leq\mathbbm{1}$. 
Now rewrite $\xi$ using the complementary basis $\cket{x}^A$ in anticipation of measuring $X^A$. The result is
\begin{align*}
\ket{\xi}^{ABMR}&=\tfrac{1}{\sqrt{d}}\sum_{x}\cket{x}^A\sum_z\omega^{xz}\ket{z}^{M}\ket{\varphi_z}^{BR}\\
&=\tfrac{1}{\sqrt{d}}\sum_{x}\cket{x}^A\left(Z^x\right)^M\sum_z\ket{z}^{M}\ket{\varphi_z}^{BR}\\
&=\tfrac{1}{\sqrt{d}}\sum_{x}\cket{x}^A\left(Z^x\right)^M\ket{\psi}^{MBR}.
\end{align*}
In the last line we have implicitly defined the state $\ket{\psi}^{MBR}$, which is just $\ket{\psi}^{ABR}$ with $A$ replaced by $M$. It is easy to work out that the result of measuring $X^A$ and marginalizing over $BM$ is the ideal output of privacy amplification of $X^A$ against $R$, namely
\begin{align*}
\bar{\xi}_X^{AR}=\tfrac{1}{d}\sum_x \cket{x}\bra{\widetilde{x}}^A\otimes \psi^R=\tfrac{1}{d}\mathbbm{1}^A\otimes \psi^R.
\end{align*}
Since $\ket{\psi}^{ABR}$ and $\ket{\psi'}^{ABMR}$ are related by the isometry $U^{B\rightarrow BM}$, measuring $X^A$ and tracing out $BM$ results in the same output for both input states. And because the fidelity cannot decrease under such a quantum operation (see Appendix~\ref{app:tdf}), this implies
\begin{align*}
F(\bar{\psi}_X^{AR},\tfrac{1}{d}\mathbbm{1}^A\otimes \psi^R)&=F(\bar{\psi'}_X^{AR},\tfrac{1}{d}\mathbbm{1}^A\otimes \psi^R)\\
&\geq F(\psi',\xi)\\
&\geq p_{\rm guess}(Z^A|R)_\psi\\
&\geq 1-\epsilon.
\end{align*}
Finally, from~\ref{eq:tdfbounds} we have $p_{\rm secure}(X^A|R)\leq \sqrt{1-F(\bar{\psi}_X^{AR},\tfrac{1}{d}\mathbbm{1}^A\otimes \psi^R)^2}\leq\sqrt{2\epsilon}$.
\end{proof}

As discussed in the previous section, there are two routes from $\epsilon$-security of $X^A$ against $R$ to $\epsilon$-recovery of $Z^A$ from $B$. The first case, when $H(X^A|B)_\psi=0$, was implicitly used by Devetak and Winter in their construction of an entanglement distillation protocol achieving the so-called hashing bound~\cite{devetak_distillation_2005}. The second case, $H(Z^A|R)_\psi=0$ has not, to our knowledge, been previously studied, but is more natural in the data compression scenario as it enforces the cq nature of the $AB$ state.
\begin{theorem}
\label{thm:secureimpliesrecover}
If $\ket{\psi}^{ABR}$ is such that $p_{\rm secure}(X^A|R)_\psi\leq \epsilon$ and either \emph{(a)} $H(X^A|B)_\psi=0$ or \emph{(b)} $H(Z^A|R)_\psi=0$, then 
$p_{\rm guess}(Z^A|B)_\psi\geq 1-\sqrt{2\epsilon}$. 
\end{theorem}
\begin{proof}
Start with case (a), whose condition implies that $\ket{\psi}^{ABR}$ takes the form
\begin{align*}
\ket{\psi}^{ABR}=\sum_x \sqrt{q_x}\cket{x}^A\cket{x}^{B_1}\ket{\vartheta_x}^{B_2R},
\end{align*}
where $B=B_1B_2$. Tracing out $B$ gives the cq state $\bar{\psi}_X^{AR}=\sum_x q_x\cket{x}\bra{\widetilde{x}}^A\otimes \vartheta_x^R$, and the condition $p_{\rm secure}(X^A|R)_\psi\leq \epsilon$ implies the fidelity of $\bar{\psi}_X^{AR}$ with the ideal output exceeds $1-\epsilon$: 
\begin{align*}
F(\bar{\psi}_X^{AR},\tfrac{1}{d}\mathbbm{1}^A\otimes\psi^R)&=\sum_x\sqrt{\tfrac{q_x}{d}}F(\vartheta_x^R,\psi^R)\\
&=\sum_x\sqrt{\tfrac{q_x}{d}}F(\ket{\vartheta_x}^{B_2R},U_x^{MB'\rightarrow B_2}\ket{\psi}^{MB'R})\\
&\geq 1-\epsilon.
\end{align*}
To get to the second line we have used Uhlmann's theorem, with corresponding isometries $U_x^{MB'\rightarrow B_2}$ for each state $\vartheta_x^R$, and the state $\ket{\psi}^{MB'R}$ is the same as $\ket{\psi}^{ABR}$ with $A$ replaced by $M$ and $B$ by $B'$. Now define the state 
\begin{align*}
\ket{\xi}^{ABR}=\tfrac{1}{\sqrt{d}}\sum_{x=0}^{d-1}\cket{x}^A\cket{x}^{B_1}U_x^{MB'\rightarrow B_2}\ket{\psi}^{MB'R},
\end{align*}
and observe that that $F(\ket{\xi}^{ABR},\ket{\psi}^{ABR})=F(\bar{\psi}_X^{AR},\tfrac{1}{d}\mathbbm{1}^A\otimes\psi^R)$. Hence $F(\ket{\xi}^{ABR},\ket{\psi}^{ABR})\geq 1-\epsilon$, and converting to trace distance, we find $D(\psi^{ABR},\xi^{ABR})\leq\sqrt{2\epsilon}$. 

Applying the conditional isometry 
\begin{align*}
V^{B_1B_2\rightarrow B_1MB'}=\sum_x\cket{x}\bra{\widetilde{x}}^{B_1}\otimes U_x^{\dagger MB'\rightarrow B_2}
\end{align*}
to $\ket{\xi}^{ABR}$ yields 
$\frac{1}{\sqrt{d}}\sum_x\cket{x}^A\cket{x}^{B_1}\ket{\psi}^{MB'R}$, and converting the result back to the $\ket{z}$ basis gives 
$\frac{1}{\sqrt{d}}\sum_z\ket{z}^A\ket{-z}^{B_1}\ket{\psi}^{MB'R}$, where aritmetic inside the state vector is modulo $d$. Thus, the measurement 
\begin{align*}
\Lambda_z^B=V^{\dagger B_1B_2\rightarrow B_1MB'}\ket{-z}\bra{-z}^{B_1}V^{B_1B_2\rightarrow B_1MB'}
\end{align*}
enables perfect recovery of $z$ from $B$ for the state $\xi^{ABR}$. But the measurement is a quantum operation, which cannot increase the trace distance, and the trace distance after a measurement is simply the variational distance of the resulting probability distributions. Therefore we can infer that
\begin{align*}
\tfrac{1}{2}\sum_{z,z'}\left|p_z\delta_{z,z'}-p_z{\rm Tr}\left[\Lambda^B_{z'}\varphi_z^B\right]\right|\leq \sqrt{2\epsilon}.
\end{align*}
Working out the lefthand side of this equation, we find that $p_{\rm guess}(Z|B)_\psi\geq 1-\sqrt{2\epsilon}$. 

Now consider case (b), whose condition implies $\ket{\psi}^{ABR}$ can be written
\begin{align*}
\ket{\psi}^{ABR}=\sum_z\sqrt{p_z}\ket{z}^A\ket{z}^{R_1}\ket{\varphi_z}^{BR_2}.
\end{align*}
Using the complementary basis for $A$ gives the alternate form
\begin{align}
\ket{\psi}^{ABR}&=\tfrac{1}{\sqrt{d}}\sum_{xz}\sqrt{p_z}\cket{x}^A\omega^{xz}\ket{z}^{R_1}\ket{\varphi_z}^{BR_2}\nonumber\\
&=\tfrac{1}{\sqrt{d}}\sum_{x}\cket{x}^A(Z^x)^{R_1}\sum_z\sqrt{p_z}\ket{z}^{R_1}\ket{\varphi_z}^{BR_2}\nonumber\\
&=\tfrac{1}{\sqrt{d}}\sum_{x}\cket{x}^A(Z^x)^{R_1}\ket{\theta}^{BR},
\label{eq:alternate}
\end{align}
where in the last line we have implicitly defined the state $\ket{\theta}^{BR}$. Observe that $\psi^R$ is invariant under the action of $(Z^x)^{R_1}$, since $\psi^R=\frac{1}{d}\sum_x (Z^x)^{R_1}\theta^R(Z^x)^{\dagger R_1}$.
Next, compute the fidelity of $\psi_X^{AR}$ with $\tfrac{1}{d}\mathbbm{1}^X\otimes\psi^R$, using the definition $\theta_x^R=(Z^x)^{R_1}\theta^R(Z^x)^{\dagger R_1}$:
\begin{align*}
F(\psi_X^{AR},\tfrac{1}{d}\mathbbm{1}^X\otimes\psi^R) 
&=\tfrac{1}{{d}}\sum_xF(\theta_x^R,\psi^R)\\
&=\tfrac{1}{{d}}\sum_x F((Z^x)^{R_1}\theta^R(Z^x)^{\dagger R_1},\psi^R)\\
&=F(\theta^R,\psi^R).
\end{align*}
Again $p_{\rm secure}(X|R)\leq \epsilon$ implies $F(\psi_X^{AR},\tfrac{1}{d}\mathbbm{1}^X\otimes\psi^R) \geq 1-\epsilon$. Since we now have $F(\theta^R,\psi^R)\geq 1-\epsilon$, it follows by Uhlmann's theorem that there exists an isometry $U^{MB\rightarrow B}$ such that $\bra{\theta}U^{MB\rightarrow B}\ket{\psi}^{MBR}\geq 1-\epsilon$. Now consider the 
state 
\begin{align*}
\ket{\xi}^{ABMR}&\equiv\tfrac{1}{\sqrt{d}}\sum_x\cket{x}^A(Z^x)^{R_1}\ket{\psi}^{MBR}\\
&=\tfrac{1}{{d}}\sum_{x,z,z'}\ket{z'}^A\omega^{x(z-z')}\sqrt{p_z}\ket{z,z}^{MR_1}\ket{\varphi_z}^{BR_2}\\
&=\sum_z\sqrt{p_z}\ket{z}^A\ket{z}^M\ket{z}^{R_1}\ket{\varphi_z}^{BR_2},
\end{align*}
from which $z$ can obviously be recovered by measuring $M$. 
The overlap of this state with $U^{\dagger MB\rightarrow B}\ket{\psi}^{ABR}$ is just $F(\theta^R,\psi^R)$, so we should expect $p_{\rm guess}(Z|B)_\psi$ to be large when using the measurement 
\begin{align*}
\Lambda^B_z=U^{MB\rightarrow B}\ket{z}\bra{z}^MU^{\dagger MB\rightarrow B}.
\end{align*}
Indeed, converting the fidelity to trace distance and working out the variational distance just as before yields $p_{\rm guess}(Z|B)_\psi\geq 1-\sqrt{2\epsilon}$.
\end{proof}

%

\section{Duality for protocols}
\label{sec:iidprotocols}
Having worked out the duality for approximate recoverability or secrecy, we can now begin investigating how the duality works for protocols designed to transform arbitrary inputs to the approximate case. Since the duality concerns transforming operations on $X^A$ into operations on $Z^A$ and vice versa, we first face the problem that operations on one necessarily affect the other in some way. By confining our analysis to PA and CSI protocols in which the outputs are linear functions of the inputs, we may avail ourselves of the stabilizer formalism, and this will enable us to ensure the back action from $X^A$ operations is consistent with the $Z^A$ transformation we wish to implement, and vice versa. A short description of those aspects of the stabilizer formalism needed here is given in Appendix~\ref{app:stabilizer}. We begin with the case of repurposing data compression into privacy amplification, as it is more straightforward. 

\begin{theorem}
Let $\mathcal{P}^\epsilon_{\rm CSI}$ be a protocol for compressing of $Z^A$ to a string $C$ of $\ell_{\rm CSI}^\epsilon$ bits via a linear compression encoding map $f:Z\rightarrow C=\{0,1\}^{\ell_{\rm CSI}^\epsilon}$. If $Z^A$ is $\epsilon$-recoverable by the decoding map $\mathcal{D}:(C,B)\rightarrow Z'$, then the encoder can be repurposed to extract $\lceil\log_2{\rm dim}(A)\rceil{-}\ell_{\rm CSI}^\epsilon$ $\sqrt{2\epsilon}$-secure bits from $X^A$ which are uncorrelated with $R$. 
\end{theorem}
\begin{proof}
First we embed system $A$ into an integer number $\lceil \log_2 {\rm dim}(A)\rceil$ of qubits. Then, using the encoding map $f$ we can define a subsystem decomposition $A=\bar{A}\hat{A}$ using $\encz=f({\bf z})$ as detailed in Appendix~\ref{app:stabilizer}. This enables us to write the input state $\ket{\Psi}^{ABR}$ as
\begin{align*}
\ket{\Psi}^{ABR}&=\sum_{\bf z}\sqrt{p_{\bf z}}\ket{\bf z}^A\ket{\varphi_{\bf z}}^{BR}\\
&=\sum_{\encz,\synz}\sqrt{p_{{\bf z}(\encz,\synz)}}\ket{\encz}^{\bar{A}}\ket{\synz}^{\hat{A}}\ket{\varphi_{{\bf z}(\encz,\synz)}}^{BR}.
\end{align*}

Since ${\bf z}$ is $\epsilon$-recoverable from the combined system $\bar{A}B$ by definition of the protocol, 
Theorem~\ref{thm:recoverimpliessecure} applies. Therefore $\synx$, the result of measuring encoded $\hat{X}$ operators on $\hat{A}$, is $\epsilon$-secure against $R$. But $\synx=g_\perp({\bf x})$, for $g_\perp$ related to $f$ as in Appendix~\ref{app:stabilizer}, so $g_\perp$ defines a key extraction function. 
As $f$ outputs a $\ell_{\rm CSI}^\epsilon$-bit string, the output of $g_\perp$ must be a string of $\lceil\log_2{\rm dim}(A)\rceil-\ell_{\rm CSI}^\epsilon$ bits.
\end{proof}  

Now we take up the converse. Again case (a) is similar to results found by Devetak and Winter in~\cite{devetak_distillation_2005}, though, because they do not use linear functions, they cannot directly interpret their use of privacy amplification as data compression of an independently-defined complementary observable. We shall return to this issue at the end of this section. Reiterating the statement made prior to Theorem~\ref{thm:secureimpliesrecover}, case (b) is more naturally suited to the data compression with side information scenario, whose input a cq state by assumption. 

\begin{theorem}
Let $\mathcal{P}^\epsilon_{\rm PA}$ be a protocol for privacy amplification of $X^A$ against $R$ consisting of a linear extraction map $g:X^A\rightarrow K=\{0,1\}^{\ell_{\rm PA}^\epsilon}$ and let the input be a pure state $\psi^{ABR}$ such that either \emph{(a)} $H(X^A|B)_\psi=0$  or \emph{(b)} $H(Z^A|R)_\psi=0$. If $\mathcal{P}^\epsilon_{\rm PA}$ produces $\ell_{\rm PA}^\epsilon$ $\epsilon$-secure bits, then the extraction map can then be used to define a compressor and corresponding decoding map which together can be used to compress $Z^A$ to $\lceil\log_2{\rm dim}(A)\rceil-\ell_{\rm PA}^\epsilon$ bits such that $Z^A$ is $\sqrt{2\epsilon}$-recoverable from the side information $B$ and compressed version $C$.
\end{theorem}
\begin{proof}
Start with case (a). Again we embed system $A$ into $d^A\equiv\lceil \log_2 {\rm dim}(A)\rceil$ bits for simplicitly. The input state has the form
\begin{align}
\ket{\Psi}^{ABR}&=\sum_{\bf x}\sqrt{q_{\bf x}}\cket{\bf x}^A \cket{\bf x}^{B_1}\ket{\vartheta_{\bf x}}^{B_2R}\nonumber\\
&=\sum_{\encx,\synx}\sqrt{q_{{\bf x}}}\ket{\encx}^{\bar{A}}\ket{\synx}^{\hat{A}}\ket{\encx}^{\bar{B_1}}\ket{\synx}^{\hat{B_1}}\ket{\vartheta_{{\bf x}}}^{B_2R},
\label{eq:secureiid}
\end{align}
where we have used the subsystem decomposition $A=\bar{A}\hat{A}$ from $\encx=g({\bf x})$ and suppressed the dependence of ${\bf x}$ on $(\encx,\synx)$.
By assumption $\encx$ is $\epsilon$-secure against $R$. Thus, Theorem~\ref{thm:secureimpliesrecover} applies to the division $\bar{A}|\hat{A}B|R$, and there exists a measurement $\Lambda^{\hat{A}B}_\encz$ such that $\encz$ is $\sqrt{2\epsilon}$-recoverable from $\hat{A}B$. 

Since $\hat{A}$ is not directly available to the decoder, we must break the measurement down into a compressor with classical output and subsequent measurement of $B$ alone, conditional on this output.  To do this, suppose $\hat{A}$ is measured in the $Z$ basis, producing $\synz$, which results in the state
\begin{align*}
\ket{\Psi_\synz}&=\sum_{\encx,\synx}\sqrt{q_{\bf x}}\ket{\encx}^{\bar{A}}\ket{\encx}^{\bar{B}_1}\omega^{\synx\cdot\synz}\ket{\synx}^{\hat{B}_1}\ket{\vartheta_{{\bf x}}}^{B_2R}\\
&=\sum_{\encx,\synx}\sqrt{q_{\bf x}}\ket{\encx}^{\bar{A}}\ket{\encx}^{\bar{B}_1}\left(X^{\synz}\right)^{\hat{B}_1}\ket{\synx}^{\hat{B}_1}\ket{\vartheta_{{\bf x}}}^{B_2R}
\end{align*}
with probability ${1}/(d^A-\ell_{\rm PA}^\epsilon)$. All $\synz$ dependence drops out when tracing out the $B$ systems, so the marginal states of $R$ conditional on $\encx$ are the same as in~(\ref{eq:secureiid}). Therefore, Theorem~\ref{thm:secureimpliesrecover} implies $\encz$ is $\epsilon$-recoverable from $B$ alone for each value of $\synz$. Since the pair $(\encz,\synz)$ fixes the value of ${\bf z}$, $\synz\equiv f_\perp({\bf z})$ is a suitable compression map enabling $\epsilon$-recovery of ${\bf z}$ from $B$ and $C=f_\perp(Z^A)$.

Now consider case (b), whose input state is of the form
\begin{align}
\ket{\Psi}&=\sum_{\bf z}\sqrt{p_{\bf z}}\ket{\bf z}^A\ket{\bf z}^{R_1}\ket{\varphi_{\bf z}}^{BR_2}\label{eq:caseb1}\\
&=\tfrac{1}{\sqrt{d^n}}\sum_{{\bf x}}\cket{\bf x}^A (Z^{\bf x})^{R_1}\ket{\Theta}^{BR_2}\nonumber\\
&=\tfrac{1}{\sqrt{d^n}}\sum_{\encx,\synx}\ket{\encx}^{\bar A}\ket{\synx}^{\hat{A}} (Z^{\encx})^{\bar{R}_1}(Z^{\synx})^{\hat{R}_1}\ket{\Theta}^{BR_2}.\nonumber
\end{align}
Here we have converted to the alternate form in the second equation, following~(\ref{eq:alternate}), with $\ket{\Theta}^{BR}=\sum_{\bf z}\sqrt{p_{\bf z}}\ket{\bf z}^{R_1}\ket{\varphi_{\bf z}}^{BR_2}$. For the third equation we again use the subsystem decomposition for $A=\bar{A}\hat{A}$ as well as $R_1=\bar{R}_1\hat{R}_1$. 
By assumption, $\encx$ is $\epsilon$-secure against $R$, so just as for case (a) Theorem~\ref{thm:secureimpliesrecover} applies to the division $\bar{A}|\hat{A}B|R$ and implies there exists a measurement $\Lambda^{\hat{A}B}_\encz$ such that $\encz$ is $\sqrt{2\epsilon}$-recoverable from $\hat{A}B$. 

This measurement can be broken down into a compression map with classical output followed by measurement of $B$ alone following the technique used in the previous case. This time, we model the measurement quantum-mechanically, as the transformation $\ket{\synz}^{\hat{A}_1}\rightarrow \ket{\synz}^{C}\ket{\synz}^{R_3}$. However, from~(\ref{eq:caseb1}) it is clear that the same effect can be achieved by the transformation $\ket{\synz}^{\hat{A}_1}\rightarrow \ket{\synz}^{C}$ followed by $\ket{\synz}^{\hat{R}_1}\rightarrow \ket{\synz}^{\hat{R}_1}\ket{\synz}^{{R}_3}$. In other words, there is no need to distribute $\synz$ to $R$, since $R$ already has a copy. Thus, the effect of the measurement is simply to transfer $\hat{A}$ to $C$. Using the function $\synz\equiv f_\perp({\bf z})$ as the compressor therefore ensures that $\encz$, and hence ${\bf z}$, is $\epsilon$-recoverable from $(B,C)$. 

In both cases $f_\perp$ outputs $d^A-\ell^\epsilon_{\rm PA}$ bits when $g$ outputs $\ell^\epsilon_{\rm PA}$, completing the proof.\end{proof}  


\section{Discussion \& Applications}
\label{sec:applications}

The reasons for restricting attention to linear hashing techniques in this analysis should now be more understandable. Since the duality between PA and CSI is meant to hold for complementary observables, it is not a priori clear that, e.g.\ a given privacy amplification function applied to $X^A$ has a well-defined action on $Z^A$, let alone the desired one. However, the use of linear hashing to deal with this problem is only shown here to be sufficient, not necessary, and it would be nice to understand more precisely under what circumstances this duality holds. 

This issue is somewhat subtle, and deserves further comment. By the results in Sec.~\ref{sec:approx}, once, say, privacy amplification of $X^A$ against $R$ has been performed, it is certainly possible to define an appropriate complementary observable $Z^A$ so that it is recoverable from $B$. However, this observable generally has nothing whatsoever to do with a complementary observable that we might have defined for the input to the privacy amplification procedure, and in particular, the two need not commute so as to be simultaneously well-defined. For instance, in~\cite{devetak_distillation_2005}, privacy amplification is used as the second step of an entanglement distillation protocol. Since the output is entangled, both $X^A$ and $Z^A$ complementary observables are recoverable from $B$. But these observables have nothing to do with complementary observables one would have defined for the \emph{input} to the protocol, so one cannot say the PA procedure performs CSI. Thus, while $\epsilon$-security of $X^A$ and $\epsilon$-recovery of $Z^A$ always go hand in hand, it does not follow from that alone that PA and CSI protocols necessarily do, too. 
On the other hand, in many situations in quantum information processing, such as in~\cite{devetak_distillation_2005}, this distinction is not important.  

Perhaps the most direct application of our results is a general entropic uncertainty relation formulated in terms of the smooth conditional min- and max-entropies~\footnote{That such a consequence ought to hold was pointed out by Matthias Christandl.}. Using the upper bound of~\ref{eq:csi}, Theorem~\ref{thm:recoverimpliessecure}, and the lower bound of~\ref{eq:pa} for an input system whose dimension $d$ is a power of two, we immediately obtain
\begin{align}
\log_2 d\leq \hmin{\sqrt{2\sqrt{2\epsilon}}}(X^A|R)_\psi+\hmax{\epsilon_1}(Z^A|B)_\psi+2\log\tfrac{1}{\epsilon_2}+4
\end{align}
for $\epsilon=\epsilon_1+\epsilon_2$. From the definition of the smoothed conditional max-entropy it follows that $\hmax{\epsilon'}(Z^A|B)\leq \hmax{\epsilon}(Z^A|B)$ for $\epsilon'\geq \epsilon$, so if we choose $\epsilon_1=\epsilon_2=\frac{\epsilon}{2}$ and $\epsilon=\frac{\delta^4}{8}$, the above expression can be transformed into the more appealing form
\begin{align}
\hmin{\delta}(X^A|R)_\psi+\hmax{\delta}(Z^A|R)_\psi\geq \log_2 d-8\log\tfrac{1}{\delta}-12.
\end{align}
This extends the recent work on uncertainty principles valid in the presence of quantum memory~\cite{renes_conjectured_2009,berta_entropic_2009} to the smooth min- and max-entropy. Due to the operational interpretations of these quantities~\cite{konig_operational_2009}, this relation should be useful in the analysis of quantum information processing protocols. 

Another application of this work is to a new approximate quantum error-correcting condition. This will be explored more fully in a future publication, but we can already give a brief overview here. Essentially, the quantum decoupling condition of~\cite{schumacher_approximate_2002} mentioned in the introduction can be broken down into two classical pieces. That condition states that $AB$ is maximally entangled when $A$ is completely uncorrelated with the purification system $R$, and it is in a completely random state. Approximate quantum error-correcting procedures can then be constructed by approximately decoupling $R$. The entanglement distillation procedure of Devetak and Winter~\cite{devetak_distillation_2005} implicitly gives a different characterization, saying that $AB$ is maximally entangled if $Z^A$ is recoverable from $B$ and secure from $R$. Using the duality of these recoverability and security notions, there are in principle two other equivalent  characterizations of approximate entanglement, from which approximate quantum error-correcting procedures can likewise be constructed. The first one states that $AB$ is maximally entangled if both $X^A$ and $Z^A$ are recoverable from $B$, a condition which was implicitly explored in~\cite{renes_physical_2008}. The second is the classical decomposition of the quantum decoupling condition, that $AB$ is entangled if both $X^A$ and $Z^A$ are secure from $R$, with the additional proviso that one of them, say $X^A$, is secure not just from $R$, but from $R$ together with a copy of $Z^A$. 


\begin{acknowledgements}
JMR acknowledges useful discussions with and careful reading of the manuscript by Mark M.\ Wilde and Matthias Christandl. Financial support was provided by the Center for Advanced Security Research Darmstadt (www.cased.de).
\end{acknowledgements}

\appendix{}

\section{Fidelity and Trace Distance}

\label{app:tdf}
Here we recount some facts about the trace distance and fidelity. Proofs can be found in, e.g.\ \cite{nielsen_quantum_2000}.
The trace distance $D(\rho,\sigma)$ between two quantum states $\rho$ and $\sigma$ is defined by
\begin{align}
\label{eq:tracedistance}
D(\rho,\sigma)=\tfrac{1}{2}\left\|\rho-\sigma\right\|_1,
\end{align}
where $\left\|A\right\|_1=\sqrt{A^\dagger A}$. 
It is invariant under unitary operations on the inputs and cannot increase under trace preserving quantum operations. In particular, if a measurement $\Lambda_k$ yields outcome $k$ with probability $r_k$ for $\rho$ and $s_k$ for $\sigma$, then the trace distance bounds the variational distance of the two distributions
\begin{align}
D(\rho,\sigma)\geq \tfrac{1}{2}\sum_k \left|r_k-s_k\right|.
\end{align}
Moreover, the trace distance is the largest probability difference the two states $\rho$ and $\sigma$ could assign to the same measurement outcome $\Lambda$, 
\begin{align}
D(\rho,\sigma)=\max_\Lambda {\rm Tr}\left[\Lambda(\rho-\sigma)\right],\quad 0\leq\Lambda\leq\mathbbm{1}.
\end{align}
Therefore, if the trace distance between $\rho$ and $\sigma$ is small, they behave nearly identically under all measurements.

Meanwhile, the fidelity $F(\rho,\sigma)$ is defined by
\begin{align}
\label{eq:fidelity}
F(\rho,\sigma)=\left\|\sqrt{\rho}\sqrt{\sigma}\right\|_1,
\end{align}
and it, too, is invariant under unitary operations on the inputs and monotonic under trace preserving quantum operations, in this case increasing. By Uhlmann's theorem the fidelity of two mixed states is related to the fidelity of their purifications. If $\ket{\psi}^{QR}$ is a purification of $\rho^Q$ and likewise $\ket{\varphi}^{QR}$ is a purification of $\sigma^Q$, then 
\begin{align}
\label{eq:uhlmann}
F(\rho,\sigma)&=\max_{U^R}F(\ket{\psi}^{QR},\left(\mathbbm{1}^Q\otimes U^R\right)\ket{\varphi}^{QR})\\
&=\max_{U^R}\bra{\psi}^{QR}\left(\mathbbm{1}^Q\otimes U^R\right)\ket{\varphi}^{QR},
\end{align}
for $U^R$ a unitary on the purifying system $R$. If this purifying system is different for the two states, say $\ket{\psi}^{QR_1}$ and $\ket{\varphi}^{QR_2}$, then the maximization is instead over partial isometries $U^{R_2\rightarrow R_1}$ taking $R_2$ to $R_1$.  

The trace distance and fidelity are essentially equivalent measures of closeness of two quantum states, via 
\begin{align}
\label{eq:tdfbounds}
1-F(\rho,\sigma)\leq D(\rho,\sigma)\leq\sqrt{1-F(\rho,\sigma)^2}.
\end{align}

\section{Smooth Entropies}
\label{app:smooth}
The smooth min- and max-entropies were first introduced by Renner \& Wolf for the classical case~\cite{renner_smooth_2004,renner_simple_2005} in order to characterize information processing protocols beyond the usual asymptotic i.i.d.\ scenario to cases where the input random variables or channels are essentially structureless. They were subsequently extended to the quantum case by Renner~\cite{renner_security_2005} and Renner \& K\"onig~\cite{renner_universally_2005}, and have undergone several additional refinements. Here we follow the definitions given in~\cite{tomamichel_duality_2009}. 

First, the conditional min-entropy for a state $\rho^{AB}$ is defined by 
\begin{align}
\hmin{}(A|B)_\rho &\equiv \max_{\sigma^B}\left(-\log \lambda_{\min}(\rho^{AB},\sigma^{B})\right),
\end{align}
with $\lambda_{\min}(\rho^{AB},\sigma^B){\equiv}\min\left\{\lambda:\rho^{AB}\leq \lambda\mathbbm{1}^A\otimes \sigma^B\right\}$. Dual to the conditional min-entropy is the conditional max-entropy, defined by 
\begin{align}
\label{eq:max}
\hmax{}(A|B)_\rho\equiv\max_{\sigma^B}\,\, 2\log F(\rho^{AB},\mathbbm{1}^A\otimes\sigma^B).
\end{align}
The two are dual in the sense that, for $\rho^{ABC}$ a pure state, $\hmax{}(A|B)_{\rho}=-\hmin{}(A|C)_{\rho}$~\cite{konig_operational_2009}.
 
Each of these entropies can be \emph{smoothed} by considering states $\bar{\rho}^{AB}$ in the $\epsilon$-neighborhood of $\rho^{AB}$, defined using the purification distance 
$P(\rho,\sigma)\equiv\sqrt{1-F(\rho,\bar{\rho})^2}$,
\begin{align}
B_\epsilon(\rho)\equiv\{\bar{\rho}:P(\rho,\sigma)\leq \epsilon\}.
\end{align}  
Note that the purification distance is essentially equivalent to the trace distance, due to the bounds 
$D(\rho,\sigma)\leq P(\rho,\sigma)\leq\sqrt{2D(\rho,\sigma)}$, which are just a reformulation of~\ref{eq:tdfbounds}.
The smoothed entropies are then given by 
\begin{align}
\hmin{\epsilon}(A|B)_\rho&\equiv\max_{\bar{\rho}\in B_\epsilon(\rho^{AB})} \hmin{}(A|B)_{\bar{\rho}},\\
\hmax{\epsilon}(A|B)_\rho&\equiv\min_{\bar{\rho}\in B_\epsilon(\rho^{AB})} \hmax{}(A|B)_{\bar{\rho}}.
\end{align}
Furthermore, the dual of $\hmax{\epsilon}(A|B)_\rho$ is $\hmin{\epsilon}(A|C)_\rho$, so that taking the dual and smoothing can be performed in either order~\cite{tomamichel_duality_2009}.  

\section{CSS Stabilizer Formalism}
\label{app:stabilizer}
The stabilizer formalism developed by Gottesman~\cite{gottesman_stabilizer_1997} in the context of quantum error-correction is perfectly suited to describing the effects of applying linear functions to the complementary observables $X^A$ and $Z^A$. In fact, here we will only need a subset of these results, for so-called Calderbank-Shor-Steane (CSS) stabilizers. Here we give an exceedingly brief overview; for more details see~\cite{nielsen_quantum_2000}. 

For simplicity, fix $d=2$; the resulting statements actually apply for any $d$ which is a power of a prime number. Starting with a collection $A$ of $n$ 2-dimensional quantum systems $A_1,\dots A_n$, suppose we would like to apply a linear function $f:\{0,1\}^n\rightarrow\{0,1\}^m$ to the result ${\bf z}$ of measuring each system $A_i$ in the $Z$ basis. Since the function is linear, each output bit is the result of computing the inner product of ${\bf z}$ with a fixed binary string ${\bf h}_j$, $f({\bf z})_j={\bf z}\cdot{\bf h}_j$. But then the $j$th output bit is nothing other than the result of measuring the operator $Z^{{\bf h}_j}\equiv Z^{h_{j,1}}\otimes \cdots\otimes Z^{h_{j,n}}$, where $h_{j,k}$ is the $k$th component of the vector ${\bf h}_j$. As much holds for $X$ by Fourier symmetry.

It can be shown that given $m$ linearly independent vectors ${\bf h}_j$, the resulting (commuting) operators $Z^{{\bf h}_j}$ \emph{stabilize} a subspace of dimension $2^k$, with $k=n-m$, meaning that there exist $2^k$ linearly independent common eigenvectors of the set. Therefore, a basis for the space $\mathbbm{F}_2^n$ translates into a basis for the space $\mathbb{C}^{2^n}$ and the operators $Z^{{\bf h}_j}$ form a complete set of commuting observables, to use language more familiar in quantum mechanics. Any basis will do, and indeed the usual decomposition of $\mathbb{C}^{2^n}$ as $n$ copies of $\mathbb{C}^2$ just corresponds to the $\mathbbm{F}_2^n$ basis of vectors defined by components $e_{j,k}=\delta_{j,k}$. 

Moreover, the algebra of $\mathbbm{F}_2^n$ carries over into the commutation relations between $X$-type stabilizers and $Z$-type stabilizers. Since $XZ=-ZX$, it follows immediately that 
\begin{align}
X^{{\bf g}_j}Z^{{\bf h}_k}=(-1)^{{\bf g}_j\cdot {\bf h}_k}Z^{{\bf h}_k}X^{{\bf g}_j}.
\end{align} 
This condition can be used to define encoded qubits and corresponding anticommuting encoded $X$ and $Z$ operators. The $X$ and $Z$ operators of the physical qubits are such that each $Z^{{\bf e}_j}$ anticommutes with just one of the $X^{{\bf e}_k}$, namely $j=k$, and commutes with all the others. This can be extended to an arbitrary basis ${\bf h}_j$ by finding its dual basis ${\bf g}_k$ for which ${\bf g}_j\cdot {\bf h}_k=\delta_{j,k}$. Each pair $({\bf g}_j,{\bf h}_j)$ then corresponds to a pair of encoded $(X,Z)$ operators. 

Suppose $f({\bf z})$ is a linear function for which the corresponding set of vectors $\{\bar{\bf h}_j\}_{j=1}^m$ is linearly independent. This set we can take as defining a basis for the corresponding subspace in $\mathbbm{F}_{2}^m$, and this basis can be completed by finding a basis of $n-m$ vectors $\{\hat{\bf h}_j\}_{j=m+1}^n$ for the complementary subspace. The complementary basis defines its own function, $f_\perp$. Since together $f$ and $f_\perp$ make up an invertible function (the associated matrix is invertible), a string ${\bf z}$ can just as well be characterized by the pair $(f({\bf z}),f_\perp({\bf z}))$. Calling these outputs $\encz=f({\bf z})$ and $\synz=f_\perp({\bf z})$, respectively, we can regard ${\bf z}$ as a function of the pair $(\encz,\synz)$. 

The stabilizer construction allows us to apply this transformation to the state vectors of the $n$ qubits as well, meaning that we can relabel the basis states $\ket{\bf z}\simeq \ket{\encz,\synz}$. That is, using the stabilizers we can perform $\mathbbm{F}_2^{n}$ arithmetic inside the kets in a meaningful way. And it respects the tensor product as well, meaning for a system $A$ of $n$ qubits we can use the collection of encoded operators for $f$ to define a subsystem $\bar{A}$ and those for $f_\perp$ to define a subsystem $\hat{A}$, so that together $\mathcal{H}_A=\mathcal{H}_{\bar{A}}\otimes\mathcal{H}_{\hat{A}}$. Or, more compactly, $A=\bar{A}\hat{A}$. 

Finally, we can now see that this formalism controls the back action from applying $f$ on $Z^A$ to the complementary operators $X^A$. In the above decomposition of $A$ into $\bar{A}$ and $\hat{A}$, we are still free to switch to the complementary basis in either subsystem, and in doing so we go from $f$ ($f_\perp$) to $g$ ($g_\perp$). If we convert $\hat{A}$ to the $X$ basis, then the resulting basis describes the possible simultaneous $f({\bf z})$ and $g_\perp({\bf x})$ outputs, even though ${\bf x}$ and ${\bf z}$ do not exist simultaneously.

\bibliographystyle{rspublicnat}
\bibliography{hswpa}
\label{lastpage}
\end{document}